\documentclass[12pt]{article}
\usepackage{amsmath,amssymb,amsthm,mathrsfs,bm,hyperref,bbding,verbatim}
\usepackage{ae,aecompl}
\newtheorem{proposition}{Proposition}[section]
\newtheorem{lemma}[proposition]{Lemma}
\newtheorem{definition}[proposition]{Definition}
\newtheorem{corollary}[proposition]{Corollary}
\newcommand{\Tr}{\protect{\text{\protect{Tr}}}}

\newcommand{\medstar}{\text{\small\FiveStarOpen}}

\numberwithin{equation}{section}

\title{Canonical \(\kappa\)-Minkowski Spacetime}

\author{Ludwik D\k abrowski\thanks{SISSA, Via Bonomea 265, 34136 Trieste, Italy},\quad Gherardo Piacitelli\({}^*\)}

\begin{document}
\maketitle

\begin{abstract}
A complete classification of the regular representations
of the relations
\[
[T,X_j]=\frac i\kappa X_j,\quad j=1,\dotsc,d
\]
is given. The quantisation of \(\mathbb R\times\mathbb R^d\) 
canonically (in the sense of Weyl) 
associated with the universal representation of the above relations
is intrinsically ``radial'', this meaning 
that it only involves
the time variable and the distance from the origin; angle variables remain
classical. 

The time axis through the origin is a spectral singularity of the
model: in the large scale limit it is
topologically disjoint from the rest.

The symbolic calculus is developed; in particular there is a 
trace functional on symbols.

For suitable choices of  states
localised very close to the 
origin, the uncertainties of all 
spacetime coordinates can be made simultaneously small at wish. On the 
contrary, uncertainty relations become important at ``large'' distances:
Planck scale effects should be visible at LHC energies, if processes
are spread in a region of size 1mm (order of peak nominal beam size) 
around the origin of spacetime.

\end{abstract}

\tableofcontents

\section{Introduction}

The \(\kappa\)-Minkowski spacetime, where the selfadjoint coordinates fulfil
\begin{equation}\label{eq:kappa_rel_gen}
[T,X_j]=\frac{i}{\kappa}X^j,\quad
[X_j,X_k]=0,\quad\quad j,k=1,\dotsc,d,
\end{equation}
has been analysed for almost 20 years \cite{{Lukierski:1991pn},Majid:1994cy}, 
mainly from the algebraic point of view. Here we will take natural units
where \(\kappa=1\), in addition to \(c=\hbar=1\).

Not much attention has been paid to representations, 
with the notable exception 
of the work of Agostini \cite{Agostini:2005mf}, 
where the representations in \(d=1\) space dimensions were constructed by 
means of the Jordan--Schwinger map, and classified by means of the theory 
of induced representations. 
In that paper (as well as in many others;
see especially \cite{Kosinski:1999dw}), Weyl operators were then defined by 
making an arbitrary choice in the order of operator products; the corresponding
quantisation then lacks the fundamental property of sending real functions
into selfadjoint operators.

Our approach is closer in spirit to that of Weyl \cite{weyl:1928}
and von Neumann \cite{vonNeumann:1931}.  After focusing on the 
appropriate regular commutation relations \`a la Weyl
\begin{equation}\label{eq:weyl_rel_intro}
e^{i\alpha T}e^{i\bm{\beta X}}=
e^{ie^{-\alpha}\bm\beta\bm X}e^{i\alpha T},
\quad \alpha\in \mathbb R,\bm\beta\in\mathbb R^d,
\end{equation}
we will show that the most general 
representation of \eqref{eq:weyl_rel_intro} is of the form
\begin{equation}
(T,X_1,\dotsc,X_d),\quad X_j=C_jR,
\end{equation}
where \(R\) is strictly positive, \(T\) has spectrum \(\mathbb R\), 
\begin{equation}\label{eq:kappa_rel_2dim}
[T,R]=iR,
\end{equation}
and each \(C_j=X_jR^{-1}\) 
is bounded, non negative, and strongly commute swith both \(T,R\).
Moreover, up to a unitary equivalence, the operators \(C_j\)
may be chosen such that 
\begin{equation}
E=\sum_iC_i^2
\end{equation} 
is an orthogonal projection. With this choice, 
\(\sum_jX_j^2=R^2E\).
In particular, the universal representation is the direct sum of a trivial 
and a non trivial component, corresponding to the two eigenspaces of \(E\).

The analogy with polar coordinates is evident; for this reason we say that
the quantisation is radial, since it only involves the commutation relations
between time and the space radius; the angular variables remain 
commutative, i.e.\ 
classical. Note that, if \(E\neq I\), then 
\(\mathbf 0\) is an isolated point of the joint spectrum
\(j\sigma(C_1,\dotsc,C_d)\); we may regard this
as the noncommutative shadow of the singularity of classical
radial coordinates in the origin of space. 
However, this singularity is not a consequence of some arbitrary choice of 
representations: it is built in the commutation relations, 
which define an  intrinsically radial model. 

The classification problem is so reduced to the case of
\(1+1\)
dimensions. The only regular irreducible representation 
of \eqref{eq:kappa_rel_2dim} with strictly positive \(R\) 
is (up to equivalence) 
\begin{equation}
T=P,\quad R=e^{-Q},
\end{equation}
where \(P,Q\) is a pair of Schr\"odinger operators, 
fulfilling \([P,Q]=-iI\). Although this was proved in \cite{Agostini:2005mf}
(under a different, though equivalent definition of regularity), 
we present a more direct, elementary proof which relies of von Neumann 
uniqueness. This proved, by Schur lemma the irreducible non trivial 
representations \((T^{(\bm c)}, \bm X^{(\bm c)})=(P,\bm ce^{-Q})\)
are labeled by vectors 
\(\bm c=(c_1,\dotsc,c_d)\in S^{d-1}=\{\bm c\in\mathbb R^d:\|\bm c\|=1\}\), while trivial
irreducible representations \((\tau,\bm 0)\), as operators on \(\mathbb C\),
are labeled by a real parameter \(\tau\).

Our subsequent 
discussion is based on the explicit computation of the 
radial Weyl operators
\begin{equation}
e^{i(\alpha T+\beta R)}=e^{i\alpha
  T}e^{i\frac{e^\alpha-1}{\alpha}
\beta R}.
\end{equation}
Indeed, together with the Weyl relations \eqref{eq:weyl_rel_intro},
they {\itshape realise} precisely 
the composition rule
described e.g.\ 
in \cite{Agostini:2002de} on the basis of the integration 
of the BCH series
done in \cite{Kosinski:1999dw}. Note that to formally apply the BCH formula to 
unbounded operators is roughly equivalent 
(through the theory of analytic vectors) to assuming that the 
representation is regular.

Contrary to the case of the CCR, where products of Weyl operators
are Weyl operators up to a phase only (the twist), the family
of radial Weyl operators \(e^{i(\alpha T+\beta R)}\) 
is a subgroup of the unitary group. Moreover, 
the correspondence with \(\mathbb R^2\) is bijective, so that \(\mathbb R^2\)
inherits a group law; the resulting group\footnote{This group plays a r\^ole
analogous to that of the Heisenberg group for the CCR. We refrain however from
calling it the \(\kappa\)-Heisenberg group because it is not a deformation
of the Heisenberg group; moreover our definition is slightly different from
the \(\kappa\)-Heisenberg group defined in \cite{Agostini:2005mf}.}~%
\(\mathfrak R\)
is isomorphic with the connected, simply connected 
Lie group whose real Lie algebra is generated by the relations
\([u,v]=-v\). Then the full Weyl operators \(e^{i(\alpha T+\bm{\beta X})}\)
form a group which is isomorphic with a 
central extension of the radial group \(\mathfrak R\). 

The natural (i.e.\ \`a la Weyl)  
prescription for the quantisation of a function \(f=f(t,\bm x)\) is
\begin{equation}\label{eq:intro_weyl_quant}
f(T,\bm X)=
\frac1{\sqrt{(2\pi)^{d+1}}}
\int\limits_{\mathbb R^{d+1}} 
d\alpha\;d\bm\beta\;\hat f(\alpha,\beta)e^{i\alpha T+\bm{\beta X}},
\end{equation}
where \(\hat f\) is the Fourier transform of 
\(f\in L^1(\mathbb R^d)\cap\widehat{L^1(\mathbb R^d)}\); we call such functions
symbols, for short.
It has all the necessary good properties: 
in particular, it sends real functions
into selfadjoint operators. 
No prescription deserves the name of quantisation
if not enjoying this property. Moreover, the 
quantisation
prescription (extended to the multiplier algebra)
sends plane waves \(e^{i(\alpha t+\bm{\beta x})}\)
precisely into the corresponding Weyl operators. 

The usual translation invariant Lebesgue 
measure \(d\alpha d\bm\beta\) which shows up
in the quantisation prescription \eqref{eq:intro_weyl_quant} is {\itshape not}
the Haar measure of the group of Weyl operators, which
is not even unimodular.
Notwithstanding this fact, 
the *-algebra of the symbols with product defined by
\begin{equation}\label{eq:star_prod}
f(T,\bm X)g(T,\bm X)=(f\star g)(T,\bm X),
\end{equation}
and pointwise conjugation as involution, reproduces precisely  the group 
algebra of the group of Weyl operators, up to (completion and) isomorphism.

From the interplay between the radial nature of the quantisation and the connection
with CCR quatisation, we obtain an unbounded linear funtional 
\begin{equation}
\tau_c(f)=\frac{2\pi^{d/2}}{\Gamma(d/2)}
\int dt\;d\bm x|\bm x|^{-d}f(t,\bm x)
\end{equation}
fulfilling
\[
\tau_c(f\star g)=\tau_c(g\star f),\quad \tau_c(\bar f\star f)\geqslant 0,
\]
which extends to an unbounded trace on the universal C*-algebra.

The universal C*-algebra of the algebra of symbols is 
\[
\mathcal A_d=\mathcal C(S^{d-1},\mathcal K)\oplus 
\mathcal C_\infty(\mathbb R), 
\]
where \(\mathcal C(S^{d-1},\mathcal K)\) is the algebra of the continuous 
functions of the sphere, with values in the algebra \(\mathcal K\)
of compact operators
on the separable, infinite dimensional Hilbert space; in other words, 
a trivial continuous field of C*-algebras on \(S^{d-1}\), with standard 
fibre \(\mathcal K\).
 
The picture is that each fibre \(\mathcal K\)
over the base point \(\bm c\in S^{d-1}\) describes the quantised open 
half plane which contains \(\bm c\) and is bounded by the time axis.
The time axis does not belong to any fibre, and remains classical: it 
is associated with the abelian  C*-algebra 
\(\mathcal C_\infty(\mathbb R)\) of continuous functions vanishing at infinity;
it arises from the trivial representations.

To say it differently, let \(\overset{\circ}{M}{}^{(1,d)}=
\mathbb R\times(\mathbb R^d-\{\bm 0\})\) be
the classical Minkowski spacetime with the time axis trhough the origin 
removed. Then \(\mathcal C(S^{d-1},\mathcal K)\) 
is to be thought of as the quantisation of 
\(\overset{\circ}{M}{}^{(1,d)}\), while \(\mathbb R\) remains classical.
This remains true for every value of \(\kappa\) and is thus bound
to survive the large scale limit \(\kappa\rightarrow\infty\).
The resulting large scale limit is indeed \(\mathbb R^{d+1}\) as a set, 
but equipped 
with a topology which makes the time axis topologically disconnected
from the rest. More precisely, it is the disjoint topological union 
\(\overset{\circ}{M}{}^{(1,d)}\sqcup\mathbb R\).

As a special case, the \(1+1\) dimensional 
\(\kappa\)-Minkowski spacetime corresponds to \(S^{0}=\{\pm 1\}\), and 
has C*-algebra
\[
\mathcal A_1=\mathcal K\oplus\mathcal C_\infty(\mathbb R)\oplus\mathcal K.
\]
In the large scale limit, the spacetime hase three disjoint 
connected components, corresponding to being on the left of the origin, 
in the origin, or on the right of the origin; see also 
\cite{Dabrowski:2009hv}.

The uncertainty relations are shortly discussed in section
\ref{sec:uncertainty}. We show that there are localisations
states which make both  the uncertainties \(\Delta T,\Delta X\)
arbitrarily small; these states are localised around the origin of space,
and at any time. This means that the 
noncommutative intrinsic limitations to localisation
arising in this model allow for 
localisation processes which in principle could transfer arbitrary high
energies to sharp localised regions of the geometric background 
by effect of localisation. This
is in plain contrast with the standard motivations underlying 
the quest of spacetime quantisation (see \cite{Doplicher:1994tu} for a particularly
careful discussion, where the probe is not implicitly assumed to posses 
spherical symmetry). On the other side, we provide  estimates on
the effects of noncommutativity at large scale.

\subsection*{Notations}
We choose natural units so that \(\kappa=1\), with the exception
of section \ref{sec:C_star}, where we discuss the large scale limit,
and of section \ref{sec:uncertainty}, where we discuss uncertainty relations. 

An important auxiliary r\^ole will be played by the Schr\"odinger
operators 
\[P=-id/ds,\quad Q=s\cdot
\] 
on their usual domain in \(L^2(\mathbb R)\).

We take the following conventions for Fourier 
transformations of functions of \(\mathbb R\times\mathbb R^d\):
\begin{subequations}
\begin{align}
\hat f(\lambda_1,\dotsc,\lambda_n)&=
(2\pi)^{-\frac{n}{2}}
\int dx_1\dotsm dx_n\;
f(x_1\dotsc,x_n)e^{-i(\lambda_1x_1+\dotsb+\lambda_nx_n)},\\
\check \varphi(x_1,\dotsc,x_n)&=
(2\pi)^{-\frac{n}{2}}
\int d\lambda_1\dotsm d\lambda_n\;
\varphi(\lambda_1\dotsc,\lambda_n)
e^{i(\lambda_1x_1+\dotsb+\lambda_nx_n)}.
\end{align}
\end{subequations}

Furthermore, we will write \(\mathscr F_j\) for the Fourier transform
in the \(j^{\text{th}}\) variable of a generic function:
\[
(\mathscr F_jf)(x_1,\dotsc,x_{j-1},\lambda,x_{j+1},\dotsc,x_n)=
\frac 1{\sqrt{2\pi}}\int dx_jf(x_1,\dotsc,x_n)e^{-i\lambda x_j},
\]
and analogously we define \(\mathscr F_{j_1,j_2,\dotsc,j_r}\),
for \(1\leqslant j_1<j_2<\dotsc<j_r\leqslant n\).

\section{Representations}\label{sec:rep}
We will begin by classifying the regular 
representations of the (formal) relations
\begin{equation}\label{eq:basic_rel}
[T,R]=iR.
\end{equation}

\begin{definition}
A pair \((T,R)\) of selfadjoint operators fulfilling
\begin{equation}\label{eq:def_reg_rep}
e^{i\alpha T}e^{i\beta R}=e^{i\beta e^{-\alpha}R}e^{i\alpha T},
\quad \alpha ,\beta \in\mathbb R,
\end{equation}
is said to fulfil \eqref{eq:basic_rel} in regular form.
\end{definition}
This definition is in agreement with a formal application of the BCH
formula. 

Let \(P,Q\) be a pair of Schr\"odinger operators. 
Then the pair \((P,e^{-Q})\)
provides a regular representation of \eqref{eq:basic_rel}, with \(R\)
positive. This can be checked directly, using that
\[
\left(e^{i\alpha P}\xi\right)(s)=\xi(s+\alpha ),\quad
\left(e^{i\beta e^{-Q}}\xi\right)(s)=e^{i\beta e^{-s}}\xi(s).
\]

A distinguished r\^ole will be played by trivial regular representations, 
namely those where the radius operator \(R\) is zero. By Schur's lemma,
irreducible trivial regular representations are one dimensional, in which case
\(T\) is a real number.


\begin{proposition}\label{prop:univ_AB}(i)
Any irreducible non trivial regular representation of \eqref{eq:basic_rel}
is equivalent to one of the following:
\begin{equation}
(T,R)=(P,\pm e^{-Q});
\end{equation}
in particular there is one only irreducible representation with \(R\) 
positive, up to equivalence.  Moreover, (ii) any trivial irreducible
representation is one dimensional, with \(T\) a real number.
\end{proposition}

\begin{proof} We already proved (ii). 
Let \((T,R)\) be a non trivial irreducible representation: 
we rewrite  the relations \eqref{eq:def_reg_rep}  in the form
\[
e^{i\alpha T}e^{i\beta R}e^{-i\alpha T}=e^{i\beta e^{-\alpha}R}.
\]
Holding \(\alpha\) fixed, the generator for the
resulting  group with parametre \(\beta\) fulfils
\begin{equation}\label{eq:generator_rel}
e^{i\alpha T}Re^{-i\alpha T}=e^{-\alpha}R.
\end{equation}
Consequently, for \(f\) a (Borel)  function of the spectrum of \(R\),
\[
f(e^{i\alpha T}Re^{-i\alpha T})=
e^{i\alpha T}f(R)e^{-i\alpha T}=f(e^{-\alpha}R).
\]
Since \((T,R)\) is irreducible and not trivial, 
\(0\) is {\itshape not} in the spectrum of \(R\) and 
we may apply the above remark to the function
\(f(x)=e^{-i\beta\log|x|}\), obtaining
\[
e^{i\alpha T}e^{i\beta(-\log|R|)}=e^{i\alpha\beta}e^{i\beta(-\log|R|)}
e^{i\alpha T},
\]
namely the Weyl relations for the CCR:
by von Neumann uniqueness \cite{vonNeumann:1931},
we may assume (up to equivalence) that \(T= P\) and
\(Q=-\log|R|\). 

Let  
\(C=\text{sign}(R)\), which commutes strongly with 
\(Q\). We rewrite again \eqref{eq:generator_rel}
in terms of \(T=P\), \(R=Ce^{-Q}\):
\[
e^{i\alpha P}Ce^{-Q}e^{-i\alpha P}=e^{-\alpha}Ce^{-Q}
\]
and, using \(e^{-Q}e^{-i\alpha P}=e^{-\alpha}e^{-i\alpha P}e^{-Q}\) and 
strict positivity of \(e^{-Q}\),
\[
e^{i\alpha P}C=Ce^{i\alpha P},
\]
namely \(C\) strongly commutes with \(P\), too. 
By the generalised Schur's lemma,
\(C\) is a multiple of the identity, \(C=\pm I\), and 
\(R= \pm e^{-Q}\).
\end{proof}

Since any real linear combination of pairwise strongly commuting selfadjoint
operators is essentially selfadjoint,  
let us introduce the following notations:
\[
\bm\beta=(\beta_1,\dotsc,\beta_d),\quad
\bm X=(X_1,\dotsc,X_d),\quad
\bm\beta\bm{X}=\left(\sum_j\beta_j X_j\right)^{**}.
\]
Moreover, \(\bm 0=(0,0,\dotsc,0)\), while 
\(\bm e_j\) denotes the usual canonical basis for \(\mathbb R^d\),
e.g. \(\bm e_d=(0,0,\dotsc,0,1)\).

The regular form of the
relations 
\begin{equation}\label{eq:more_dim_rels}
[T,X_j]=iX_j,\quad [X_j,X_k]=0
\end{equation}
among the selfadjoint operators \(T,X_1,\dotsc,X_d\) can be easily 
generalised:
\begin{definition} A set \((T,\bm X)\) of \(d+1\) selfadjoint
operators fulfilling
\begin{subequations}
\label{eq:genrel}
\begin{gather}
\label{eq:genrel_main}
e^{i\alpha T}e^{i\bm{\beta X}}=
e^{ie^{-\alpha}\bm{\beta X}}e^{i\alpha T},\quad
\quad \alpha\in \mathbb R,\bm\beta\in\mathbb R^d,\\
[e^{i\bm{\beta X}},e^{i\bm{\beta' X}}]=0,\quad
\quad \bm\beta,\bm\beta'\in\mathbb R^d,
\end{gather}
\end{subequations}
is said a regular representation of the relations \eqref{eq:more_dim_rels}.
\end{definition}

We will call a representation \((T,\bm X)\) trivial
if \(\bm X=0\). Irreducible trivial representations are then of the form
\(T=t,X_1=\dotsc =X_d=0\) as operators on the one dimensional Hilbert space
\(\mathbb C\). It is convenient to introduce the notation 
\begin{equation}
T^{(0)}=P,\bm X^{(0)}=\bm 0.
\end{equation}
Since
\(P\) may be replaced with \(Q\) by means of a canonical transformation, the
trivial representation \((T^{(0)},\bm X^{(0)})\) contains every irreducible 
trivial representation precisely once (up to equivalence).

To every \( \bm c=( c_1,\dotsc, c_d)
\in\mathbb R^d\backslash\{0\}\) there
is an irreducible regular representation
\[
T^{(\bm c)}=P,\quad X_j^{(\bm c)}= c_je^{-Q}
\]
where \(P,Q\) are the Schr\"odinger operators on \(L^2(\mathbb R)\).
It is clear that, since the length of \(\bm c\) may be rescaled
by means of a canonical transformation (see the proof below), 
we may restrict ourselves to \(|\bm c|=1\);
then, different choices of \(\bm c\) give inequivalent representations.
In other words, there is a family of pairwise inequivalent
representations labeled by the unit sphere \(S^{d-1}=
\{ \bm c\in\mathbb R^d:|\bm c|=1\}\). 

We now prove that there are no other irreducible representations, 
up to unitary equivalence.

\begin{proposition}
Let \((T,\bm X)\) be a non trivial irreducible regular representation of the
relations \eqref{eq:genrel}; 
then there is a unique \(\bm c\in S^{d-1}\) such that 
\begin{equation}
(T,\bm X)\simeq
(T^{(\bm c)},\bm X^{(\bm c)}).
\end{equation}
Moreover, the representation \((T^{(\bm 0)},\bm X^{(\bm 0)})\) contains
a representative for each class of trivial irreducible representations, 
without multiplicity.
\end{proposition}

\begin{proof}
With \(R^2=(X_1^2+\dotsc+X_d^2)^{**}\), let \(R=\int_0^{\infty} r\;dF(r)\) 
be the spectral resolution of \(R\), and \(E=F(\infty)-F(0)=
\chi_{(0,\infty)}(R)\), where \(\chi_{(0,\infty)}\) is the characteristic
function of the set \((0,\infty)\). By \eqref{eq:generator_rel}
\[
e^{i\alpha T}Ee^{-i\alpha T}=\chi_{(0,\infty)}(e^{-\alpha} R)=E,
\]
so that \(E\) commutes strongly both with \(R\) and \(T\). Hence, by the 
generalised Schur's lemma, either \(E=I\)
or \(E=0\). If \(E=0\), the representation is trivial: \(T\) 
is multiplication by a real number on a one dimensional Hilbert space 
and \(\bm X=\bm C=0\).
Otherwise \(R\) is invertible, and the  bounded operators
\[
C_k=X_kR^{-1},\quad k=1,\dotsc,d,
\]
strongly commute pairwise and with \(R\). By \eqref{eq:generator_rel}
and the properties of functional calculus,
\[
e^{i\alpha T}C_ke^{-i\alpha T}=
e^{i\alpha T}X_ke^{-i\alpha T}e^{i\alpha T}R^{-1}e^{-i\alpha T}=
(e^{-\alpha}X_k)(e^{-\alpha}R)^{-1}=C_k,
\]
so that \(C_k\) strongly commutes with \(T\), too.
Hence by Schur's lemma \(C_k=c_kI\) for some \(c_k\), and \(X_k=c_kR\). 

Since the representation is not trivial, there is at least some \(c_j\neq 0\):
thus \eqref{eq:genrel_main} written for \(\bm\beta=\beta\bm e_j\) 
gives precisely \eqref{eq:def_reg_rep}. It follows that (up to equivalence)
\(T=P\), \(R=e^{-Q}\), by positivity of \(R\) and 
proposition \ref{prop:univ_AB}. 

With \(|\bm c|^2=\sum_j|c_j|^2\), the unitary
operator \(U=e^{iP\log|\bm c|}\) fulfils
\(
Ue^{-Q}U^*=(1/|\bm c|)e^{-Q},
\)
so that we may assume \(\bm c\in S^{d-1}\), the unit sphere.
\end{proof}

Note that, for \(d=1\), \(S^0=\{\pm 1\}\), so that there are two only
equivalence classes of non trivial, irreducible representations: 
\((P,\pm e^{-Q})\). We thus recovered the special case discussed in 
\cite{Agostini:2005mf,Dabrowski:2009hv}.

\begin{definition}\label{def:univ_rep}
Let \(d\bm c\) be the rotation invariant Lebesgue measure on 
\(S^{(d-1)}\). 
Let \(d\mu(\bm c)\) be the measure on
\[
S^{d-1}\sqcup\{\bm 0\},
\]
defined by \(\int d\mu(\bm c)f(\bm c)=f(\bm 0)+\int d\bm cf(\bm c)\). The 
universal representation of the relations \eqref{eq:genrel} is
%
\begin{subequations}
\begin{align}
T^u&= I\otimes P,\\ 
R^u&= I\otimes e^{-Q},\\
C_j^u&= (c_j\cdot)\otimes I,\\ 
X_j^u&= C_j^uR^u=(c_j\cdot)\otimes e^{-Q},
\end{align}
on the Hilbert space 
\begin{equation}
\mathfrak H^u=
L^2(S^{d-1}\sqcup\{\bm 0\},d\mu(\bm c))\otimes L^2(\mathbb R),
\end{equation}
\end{subequations}
where \(c_j\cdot\) is the operator of multiplication by \(c_j\).
\end{definition}
By construction, the above representation contains precisely one 
representative for every equivalence class of irreducible representations;
for this reason we called it universal. By taking  a suitable amplification,
we can easily obtain a representation which is covariant under a unitary 
representation of the group \(G_d\)
of orthogonal space transformations, time translations, and space dilations.
\begin{proposition}\label{prop:cov_coord}
Let \(G_d=O(\mathbb R^d)\times\mathbb R\times(0,\infty)\) be the Kronecker 
product of the orthogonal group, the additive group \(\mathbb R\) and
the multiplicative group \((0,\infty)\), so that 
\begin{equation}
(A_1,a_1,\lambda_1)(A_2,a_2,\lambda_2)=(A_1A_2,a_1+\lambda_1a_2,\lambda_1\lambda_2),
\quad (A_j,a_j,\lambda_j)\in G_d.
\end{equation}
There exists a \(G_d\)-covariant representation, namely a strongly 
continuous unitary representation \(U\) of \(G_d\) and a representation
\((T,\bm X)\) such that, for every 
\((A=(a_{jk}),a,\lambda)\in G_d\), 
\begin{subequations}
\begin{align}
  U(A,a,\lambda)^{-1}  X_j  U(A,a,\lambda)&=\lambda\sum_ka_{jk}  X_k,\quad j=1,\dotsc,k,\\
  U(A,a,\lambda)^{-1}  T  U(A,a,\lambda)&=  T+aI.
\end{align}
\end{subequations}
\end{proposition}
\begin{proof}
Take
\begin{align*}
T&=I\otimes P\otimes I,\\
X_j&=C_j\otimes e^{-Q}\otimes e^{-Q},\\
\end{align*}
on \(\mathfrak H^u\otimes L^2(\mathbb R)\), and
\[
U(A,a,\lambda)\eta\otimes\xi\otimes\xi'=\eta(A^{-1}\cdot)\otimes e^{iaQ}\otimes e^{i(\log\lambda)P}.
\]
\end{proof}

\section{Radial Weyl Operators, Quantisation and Trace}
\label{sec:weyl}

According to the discussion of section \ref{sec:rep}, the 
quantisation will take place (in the non trivial component)
in the radial directions labeled 
by vectors \(\bm c\in S^{d-1}\). Hence, we discuss preliminarly
the quantisation corresponding to the operators \(T=P,R=e^{-Q}\).
Our first task is to compute the Weyl operators \(e^{i(\alpha T+\beta R)}\).

\begin{proposition}\label{prop:eureka}
Let \(T,R\) be selfadjoint operators on some Hilbert space,
fulfilling \eqref{eq:def_reg_rep} with \(R>0\). 
Then, for every \((\alpha,\beta)\in\mathbb R^2\),
the operator \(\alpha T+\beta R\) is essentially
selfadjoint and fulfils 
\begin{equation}
e^{i\alpha T+\beta R}=
e^{i\alpha T}e^{i\frac{e^{\alpha} -1}{\alpha }\beta R}.
\end{equation}
\end{proposition}

Once the right ansatz has been guessed, the proof is a standard application of
the  Stone--von Neumann theorem \cite[theorem VIII.10]{reed-simon_1}, 
which we omit. Indeed,
we find it more instructive to describe a method for finding the 
right ansatz, which does not rely on a formal application of the BCH formula. 


Assume that there there is a common dense domain \(\mathscr X\) on which
\(\alpha T+\beta X\) is essentially selfadjoint for every \(\alpha,\beta\).
The operators \(W(\alpha,\beta):=e^{i\alpha T+\beta X}\) should fulfill the
following properties:
\begin{gather}
\label{eq:W_1_groups}W(\alpha ,0)=e^{i\alpha T},\quad W(0,\beta )=e^{i\beta X},\\
\label{eq:W_unitarity} W(\alpha ,\beta )^{-1}=W(\alpha ,\beta )^*,\\
\label{eq:W_1_group} W(\lambda \alpha ,\lambda \beta )W({\lambda'} \alpha ,{\lambda'} \beta )=
W((\lambda+{\lambda'})\alpha ,(\lambda+{\lambda'}) \beta )
\end{gather}
identically for \(\alpha ,\beta ,\lambda,{\lambda'}\in\mathbb R\).
To solve the above problem, we took 
the following ansatz:
\[
W(\alpha ,\beta )=e^{ir(\alpha ,\beta )T}e^{is(\alpha ,\beta )X}.
\]
Some little effort leads to the given solution.

We now discuss the properties of the map
\[
f\mapsto f(T,R)=\frac1{2\pi}\int d\alpha\;d\beta\;\hat f(\alpha,\beta)
e^{i(\alpha T+\beta R)},
\]
defined on the class 
\(L^1(\mathbb R^2)\cap\widehat{L^1(\mathbb R^2)}\) of symbols, where
\[
\hat f(\alpha,\beta)=\frac1{2\pi}\int dt\;dr\;f(t,r)e^{-i(\alpha t+\beta r)}.
\]
In this section we always will take \(T=P,R=e^{-Q}\) on \(L^2(\mathbb R)\). 

With the explicit action 
\[
(e^{i(\alpha P+\beta e^{-Q})}\xi)(s)=e^{i\frac{1-e^{-\alpha}}{\alpha}\beta e^{-s}}
\xi(s+\alpha),
\]
standard computations yield
\[
(f(T,R)\xi)(s)=\int du\;K_f(s,u)\xi(u),
\]
where
\begin{equation}\label{eq:kernel_kappa}
K_f(s,u)=(\mathscr F_1 f)\left(u-s,\frac{e^{-s}-e^{-u}}{u-s}\right);
\end{equation}
here above, \(\mathscr F_1\) denotes the Fourier transform 
in the first variable (for conventions, see the end of the introduction). 

Inspection of \eqref{eq:kernel_kappa} gives immediately  
\begin{lemma}\label{lm:only_right_of_f}
Let \(f_1,f_2\in L^1(\mathbb R^2)\cap\widehat{L^1(\mathbb R^2)}\). We have
\[
f_1(T,R)=f_2(T,R)
\]
if and only if
\[
f_1(t,r)=f_2(t,r),\quad t\in\mathbb R, r\in (0,\infty).
\]
\end{lemma}
Equivalently, the restriction of the map \(f\mapsto f(T,R)\) to the symbols
\(f\) which are even in the second variable (namely 
\(f(t,\cdot)= f(t,|\cdot|)\) for all \(t\)'s) is injective.

Moreover,
\begin{lemma}\label{lm:HS}
Let \(f\in L^1(\mathbb R^2)\cap\widehat{L^1(\mathbb R^2)}\). If
\[
\int\limits_{r>0} dt\;dr\;\frac 1r|f(t,r)|^2<\infty,
\]
then the operator \(f(T,R)\) is Hilbert-Schmidt, with Schatten norm
\begin{equation}
\|f(T,R)\|_2=\left(\;\int\limits_{r>0} dt\;dr\;\frac 1r 
\left|f(t,r)\right|^2\right)^{1/2}.
\end{equation}
\end{lemma}
\begin{proof}
By lemma \ref{lm:only_right_of_f}, we may assume \(f(t,r)=f(t,-r)\) 
identically, without loss of generality. With the substitution
\[
\alpha=u-s,\quad r=\frac{e^{-s}-e^{-u}}{u-s},
\]
we find
\begin{align*}
\|K_f\|^2_{L^2(\mathbb R^2)}&=
\int du\;ds\; \left|(\mathscr F_1 f)\left(u-s,\frac{e^{-s}-e^{-u}}{u-s}\right)\right|^2=\\
&=\int\limits_{r>0} d\alpha\;dr\;\frac 1r 
\left|(\mathscr F_1 f)(\alpha,r)\right|^2=\int\limits_{r>0} dt\;dr\;\frac 1r 
\left|f(t,r)\right|^2.
\end{align*}
where we used that \(\mathscr F_1\) is unitary on 
\(L^2(\mathbb R^2,|r|^{-1}dt\;dr)\). The result then follows from classical
theorems (see e.g.\ \cite[Theorem VI.23]{reed-simon_1}).
\end{proof}

We next ask ourselves when \(f(T,R)\) has trace, and how to compute it.
We will give a somewhat indirect argument, which is of some
interest on its own. 

In some sense, by its very definition the operator \(f(T,R)\) 
appears as a ``function'' of the Schr\"odinger operators \(P,Q\). 
Hence, it is natural to expect (at least for a suitable subclass of 
symbols) 
that there exists a function \(g\) such that 
\(f(T,R)=g(P,Q)\), 
where the latter is intended as the CCR--Weyl quantisation
\[
g(P,Q)=\frac{1}{2\pi}
\int d\alpha\;d\beta\; \hat g(\alpha,\beta)e^{i(\alpha P+\beta Q)}.
\] 
Such a map\(f \mapsto g\) would allow for computing the trace of \(f(T,R)\), by known results on Weyl quantisation\footnote{Note that it would also allow for
describing the star product of \(\kappa\)-Minkowski symbols as the pull-back
ot the twisted product defined by the Weyl quantisation associated with CCR;
the relations with the approach of \cite{GraciaBondia:2001ct} (see especially 
eq.\ (5.1) therein) will be discussed elsewhere.}. 

Indeed it is well known 
(see e.g.\ \cite[4.1, eq. (59)]{stein:1993aaa}) that
\[
(g(P,Q)\xi)(s)=\int dt\;H_g(s,u)\xi(u),
\]
where
\begin{equation}
H_g(s,u)=(\mathscr F_1g)\left(u-s,\frac{u+s}2\right).
\end{equation}
The operators \(g(P,Q),f(T,R)\) are the same if and only if
they have the same integral kernel (a.e.).
Setting \(\lambda=u-s,q=(u+s)/2\), from the condition
\begin{equation}\label{eq:ker_cond}
H_g(s,u)\equiv K_f(s,u)
\end{equation}
we get
\begin{equation}\label{eq:same_kernels}
(\mathscr F_1g)(\lambda,q)\equiv
(\mathscr F_1f)\left(\lambda,e^{-q}\frac{\sinh(\lambda/2)}{\lambda/2}
\right).
\end{equation}

\begin{proposition}\label{prop:irred_trace}
For any \(f\in L^1(\mathbb R^2)\cap\widehat{L^1(\mathbb R^2)}\) 
fulfilling
\[
\int\limits_{\{r>0\}} dt\;dr\;\frac{1}{r}|f(t,r)|<\infty,
\]
the operator \(f(T,R)\) is trace class, with trace
\begin{equation}
\Tr(f(T,R)=\int\limits_{\{r>0\}} dt\;dr\;\frac{1}{r}f(t,r).
\end{equation}
\end{proposition}
\begin{proof} Use
\eqref{eq:same_kernels} and \(\Tr(g(P,Q))=\int dp\;dq\;g(p,q)\).
\end{proof}

\section{The Radial Algebra}

While in  the canonical case the product of Weyl operators is a Weyl operator
only up to a phase (called the twist), 
in this case the product of two radial 
Weyl operators \(e^{i(\alpha P+\beta e^{-Q})}\)
is again a radial Weyl operator. The reason is that in the case of 
\(\kappa\)-Minkowski the identity is not involved in the 
commutation relations. Hence the set
of Weyl operators is a subgroup of the unitary group. Since
the correspondence between \(\mathbb R^2\) and the radial Weyl operators
is one to one, \(\mathbb R^2\) inherits from the composition rule of radial
Weyl  operators a group law:

\begin{definition}The (locally compact) group \(\mathfrak R\)
is the 
group obtained by equipping \(\mathbb R^2\)
with the usual topology and
the group law
\[
(\alpha_1,\beta_1)(\alpha_2,\beta_2)=(\alpha_1+\alpha_2,
w(\alpha_1+\alpha_2,\alpha_1)e^{\alpha_2}\beta_1
+w(\alpha_1+\alpha_2,\alpha_2)\beta_2),
\]
where
\begin{equation}\label{eq:star_prod_aux}
w(\alpha,\alpha')=\frac{\alpha(e^{\alpha'}-1)}%
{\alpha'(e^{\alpha}-1)}.
\end{equation}
\end{definition}

Note that \(w\) (which is always understood to be extended
to the full \(\mathbb R^2\) by continuity) fulfills 
\begin{equation}
w(\alpha_1,\alpha_2)w(\alpha_2,\alpha_3)=
w(\alpha_1,\alpha_3),\quad w(\alpha_1,\alpha_2)e^{\alpha_1}
=-w(-\alpha_1,\alpha_2)
\end{equation}
identically, and is always positive. 
The identity of \(\mathfrak R\) is \((0,0)\), and 
\((\alpha,\beta)^{-1}=(-\alpha,-\beta)\). Finally
\begin{equation}
\lim_{\kappa\rightarrow\pm\infty}w(\alpha_1/\kappa,\alpha_2/\kappa)=1,
\end{equation}
so that we may regard \(\mathfrak R\) as a non abelian
deformation ot the additive group
\((\mathbb R^2,+)\).

As anticipated, the above definition finds its motivations in the following 
\begin{lemma}
The map
\[
\mathfrak R\ni (\alpha,\beta)\mapsto e^{i(\alpha P+\beta e^{-Q})}
\]
is a strongly continuous\footnote{Note however that, 
from proposition \ref{prop:eureka} and \cite{singer}, 
it follows that a nontrivial representation   
cannot be continuous in the operator norm topology.}, faithful
unitary representation of the radial group.
\end{lemma}
The proof consists of a direct check. The group \(\mathfrak R\) and its Lie algebra are well known:
\begin{proposition}
(i) The generators \(u,v\) of 
the real Lie algebra \(\mathfrak r=\text{Lie}(\mathfrak R)\), corresponding
to  the one--parameter subgroups \(\lambda\mapsto(\lambda,0)\)
and \(\lambda\mapsto(0,\lambda)\) respectively, fulfil
\[
[u,v]=-v.
\]
(ii) \(\mathfrak R\) is isomorphic to the subgroup\footnote{The group \(G\) is a variant
of the so called ``\(ax+b\)'' group.} 
\begin{equation}
G=\left\{\begin{pmatrix}e^{a}&0\\b&1\end{pmatrix}:
(a,b)\in\mathbb R^2\right\}
\end{equation}
of \(GL(\mathbb R^2)\); an explicit isomorphism 
\(j:\mathfrak R\rightarrow G\) is given by
\begin{equation}
j((\alpha,\beta))=
\begin{pmatrix}e^{\alpha}&0\\
\frac{e^{\alpha}-1}{\alpha}\beta&1\end{pmatrix}.
\end{equation}
\end{proposition}

\begin{proof}The proof of (ii) is a direct check. 
Then (i) follows immediately, since
connected, simply connected  Lie groups are isomorphic if and only if
their Lie algebras also are isomorphic.
\end{proof}

Note that the choice of the generators \(u,v\) is such that
for any strongly continuous unitary representation \(W\) of \(\mathfrak R\), 
the selfadjoint operators \(T,R\) defined by
\(W(\text{Exp}\{\lambda u\})=e^{i\lambda T}\),
\(W(\text{Exp}\{\lambda v\})=e^{i\lambda R}\) are a regular representation
of the relations \([T,R]=iR\), where \(\text{Exp}\) is the 
Lie exponential map.

Note also that under the isomorphism \(j\) of (i) above, 
the one parameter subgroups
\(\lambda\mapsto(\lambda\alpha,\lambda\beta)\) of \(\mathfrak R\) 
are mapped  into
\begin{equation}
\lambda\mapsto
\begin{pmatrix}e^{\lambda\alpha}&0\\
\frac{e^{\lambda\alpha}-1}{\alpha}\beta&1\end{pmatrix}.
\end{equation}
This should completely clarify the relations of our definition of 
Weyl operators with other (non canonical) definitions available in the 
literature, which also are
related with such Lie groups, but do not fulfill the essential condition
that \(\lambda\mapsto W(\lambda\alpha,\lambda\beta)\) is a one parameter
group. See e.g.\ \cite{Kosinski:1999dw,Agostini:2005mf}, where the Weyl operators are defined
with the choice ``time first'': \(e^{i\alpha T}e^{i\alpha R}\).

The group \(\mathfrak R\) is not unimodular. Indeed
\begin{lemma}\label{lm:haar}
(i) The left Haar measure and modular function 
of the group \(\mathfrak R\)  
are
\begin{equation}
d\mu(\alpha,\beta)=\frac{e^{\alpha}-1}{\alpha}d\alpha\;d\beta,
\quad \Delta(\alpha,\beta)=e^{\alpha}.
\end{equation}

(ii) With the above choice of the normalisation, for
every bounded function \(\varphi\) with compact support,
\begin{equation}
\lim_{\kappa\rightarrow\infty}
\kappa^2\int d\mu(\alpha,\beta)\varphi(\alpha\kappa,\beta\kappa)=
\int d\alpha\;d\beta\;\varphi(\alpha,\beta).
\end{equation}
\end{lemma}

\begin{proof}
The proof of (i) is a direct check.
(ii) follows from the fact that the function 
(\(\alpha\mapsto\kappa(e^{\alpha} -1)/\alpha\)) converges to 
(\(\alpha\mapsto 1\)) uniformly on compact sets as \(\kappa\rightarrow\infty\).
\end{proof}

We describe explicitly in our case some basic facts of 
abstract harmonic analysis (see e.g. \cite[\S\S29-30]{loomis}):
\begin{enumerate}
\item 
Let \(\mu,\Delta\) be the left Haar measure and modular function, respectively,
described in lemma \ref{lm:haar}.
The group *-algebra \(L^1(\mathfrak R)\) of \(\mathfrak R\) is obtained
by equipping the Banach space \(L^1( \mathbb R^2,\mu)\) 
with the (convolution) product
\begin{subequations}
\begin{align}\nonumber
(\varphi\times\psi)(\alpha,\beta)=&\int d\mu(\alpha',\beta')
\varphi(\alpha',\beta')\psi((\alpha',\beta')^{-1}(\alpha,\beta))=\\
=&\nonumber\int
d\alpha'\;d\beta'\;
\frac{e^{\alpha'}-1}{\alpha'}
\varphi(\alpha',\beta')\\
&\psi(\alpha-\alpha',w(\alpha-\alpha',\alpha)\beta-
w(\alpha-\alpha',-\alpha')e^{\alpha}\beta')
\end{align}
and the involution
\begin{equation}
\varphi^\dagger(\alpha,\beta):=\Delta((\alpha,\beta)^{-1})\overline{\varphi((\alpha,\beta)^{-1})}=
e^{-\alpha}\overline{\varphi(-\alpha,-\beta)};
\end{equation}
\end{subequations}
the group algebra is a Banach *-algebra, namely the product is continuous and
\[
\|\varphi^\dagger\|_{L^1(\mathfrak R)}=\|\varphi\|_{L^1(\mathfrak R)},\quad(\varphi\times\psi)^\dagger=\psi^\dagger\times\varphi^\dagger,
\]
where 
\[
\|\varphi\|_{L^1(\mathfrak R)}=\int d\mu(\alpha,\beta)|\varphi(\alpha,\beta)|.
\]
\item
Let \(W\) be a unitary representation of the group \(\mathfrak R\);
then
\[
\Pi(\varphi):=\int d\mu(\alpha,\beta)\varphi
(\alpha,\beta)W(\alpha,\beta),
\quad \varphi\in L^1(\mathfrak R),
\]
defines a *-representation of the group algebra \(L^1(\mathfrak R)\):
\begin{gather*}
\Pi(\varphi)\Pi(\psi)=\Pi(\varphi\times\psi),
\quad\varphi,\psi\in L^1(\mathfrak R),\\
\Pi(\varphi^\dagger)=\Pi(\varphi)^*;
\end{gather*}
the preservation of involution is a consequence of unitarity of \(W\).
Since \(\|W(\alpha,\beta)\|=1\), 
\[
\|\Pi(\varphi)\|\leqslant \|\varphi\|_{L^1(\mathfrak R)}
\]
and the representation is continuous (actually, any *-representation of
a Banach *-algebra by bounded operators on a Hilbert space is continuous 
for general reasons, see e.g.\ \cite[1.3.7]{dixmierC}).
\end{enumerate}

\begin{definition}
The algebra \(\mathscr B\) is the Banach
*-algebra obtained by equipping \(L^1(\mathbb R^2)\) with the product
\begin{subequations}
\begin{align}\nonumber
(\varphi_1\ast\varphi_2)(\alpha,\beta)
=\int&  d\alpha'd\beta'w(\alpha-\alpha',\alpha)\varphi_1(\alpha',\beta')\\&
\varphi_2(\alpha-\alpha',w(\alpha-\alpha',\alpha)\beta-
w(\alpha'-\alpha,\alpha')\beta');
\label{eq:def_starprod_momsp}
\end{align}
and the involution
\begin{equation}
\varphi^*(\alpha,\beta)=\overline{\varphi(-\alpha,-\beta)}.
\end{equation}
\end{subequations}
\end{definition} 
The above definition is motivated by the following lemma, which also proves 
that it is well posed.
\begin{lemma}\label{lm:radial_algebra_group} 
The Banach *-algebra \(\mathscr B\) and the group algebra
\(L^1(\mathfrak R)\) are isomorphic. 
With the usual normalisation of the Lebesgue measure on \(\mathbb R^2\) and the normalisation of the Haar measure described in lemma \ref{lm:haar},
the isomorphism can be chosen isometric.
\end{lemma}
\begin{proof}
The linear map
\[
u:L^1(\mathbb R^2)\rightarrow L^1(\mathfrak R)
\]
defined by
\[
(u\varphi)(\alpha,\beta)=\frac{\alpha}{e^\alpha-1}\varphi(\alpha,\beta)
\]
is evidently an equivalence of Banach spaces. The proof that
\[
(u\varphi)\times(u\psi)=u(\varphi\ast\psi),\quad u(\varphi^*)=
u(\varphi)^\dagger
\]
for \(\phi,\psi\in L^1(\mathfrak R)\) consists of a direct check.
\end{proof}

\begin{definition}
A representation \(\pi\) of \(\mathscr B\) is said trivial if
\(\pi(\mathscr B)\subset \pi(\mathscr B)'\).
\end{definition}

\begin{proposition}\label{prop:bij_group_alg_rep}
Let \(W\) be a strongly continuous unitary representation of the
radial group \(\mathfrak R\). Then
\[
\pi(\varphi)=\int d\alpha\;d\beta\;\varphi(\alpha,\beta)W(\alpha,\beta),\quad
\varphi\in L^1(\mathbb R^2)
\] 
defines a *-representation \(\pi\) of the *-algebra \(\mathscr B\).
This establishes a bijection between equivalence classes of strongly continuous
unitary representations of the group \(\mathfrak R\) and the
equivalence classes of *-representations of the *-algebra \(\mathscr B\). 
In particular, a representation \(\pi\) of \(\mathscr B\) 
is trivial if and only if the corresponding representation \(W\) of 
\(\mathfrak R\) gives a trivial regular representation of the Weyl relations.

Any trivial irreducible representation is equivalent, 
for some \(t\in\mathbb R\), to the one dimensional
representation \(\pi(\varphi)=\check\varphi(t,0)\).

Any non trivial irreducible representation is equivalent to one of
the representations \(\pi_\pm\), where 
\[
\pi_{\pm}(\varphi)=\int d\alpha\; d\beta\;\varphi(\alpha,\beta)
e^{i(\alpha P\pm \beta e^{-Q})}
\]
on \(L^2(\mathbb R)\).
Note that
\[
\pi_+(\hat f)=f(T,R).
\]
\end{proposition}
The proof is an immediate consequence
of the remark that \(\pi\circ u=\Pi\), where
\(\Pi\) is the representation of the group algebra associated with \(W\).

We are now ready to give the following, crucial
\begin{definition}
The radial algebra is the algebra \(\mathscr R\subset\mathscr B\) of
fixed elements under the automorphism 
\(\digamma:\mathscr B\rightarrow \mathscr B\), where 
\((\digamma\varphi)(\alpha,\beta)=\varphi(\alpha,-\beta)\). 
\end{definition}

It is clear that the representations \(\pi_\pm\) of \(\mathscr B\),
described in the above corollary, are related by \(\pi_-=\pi_+\circ\digamma\).
It follows that their restrictions 
\begin{equation}\label{eq:pi_radial}
\pi_r=\pi_\pm\restriction_{\mathscr R}
\end{equation}
to the radial algebra coincide. Indeed, there
are no other irreducible (classes of) irreducible representation.
\begin{proposition}\label{radial_alg_rep}
Every non trivial 
irreducible representation of the radial algebra is equivalent to
\(\pi_r=\pi_+\restriction_{\mathscr R}\).
\end{proposition}
\begin{proof}
The representation \(\pi_r\) has the same image as any of \(\pi_\pm\),
hence it is clearly irreducible. Conversely, let us first show that
the operators \(E_\pm\) on \(L^1(\mathbb R^2)\) defined by
\begin{equation}
(E_\pm\varphi)\check{\phantom |}(t,x)=\check\varphi(t,\pm|x|).
\end{equation}
are *-homomorphism from \(\mathscr B\) onto
its Banach *-subalgebra \(\mathscr R\). Surjectivity is obvious. 
To prove multiplicativity, namely 
\begin{equation}
E_\pm(\varphi\ast\psi)=(E_\pm\varphi)\ast(E_\pm\psi),
\end{equation}
we observe that, by the same argument of lemma \ref{lm:only_right_of_f}
(or using that \(f(T,R)=\pi_+(\hat f)\)), 
\(\pi_\pm\) fulfils \(\pi_\pm\circ E_\pm=\pi_\pm\); hence, 
\begin{align*}
\pi_r(E_\pm(\varphi\ast\psi))&=\pi_\pm(\varphi\ast\psi)=
\pi_\pm(\varphi)\pi_\pm(\psi)=\pi_\pm(E_\pm\varphi)\pi_\pm(E_\pm\psi)=\\
&=\pi_\pm((E_\pm\varphi)\ast(E_\pm\psi))=\pi_r((E_\pm\varphi)\ast(E_\pm\psi)),
\end{align*}
where in the last step we used that \((E_\pm\varphi)\ast(E_\pm\psi)\in\mathscr R\).
Again by the argument of lemma \ref{lm:only_right_of_f}, \(\pi_r\) is injective,
so that multiplicativity is proved. Finally, by 
\((\varphi^*)\check{\phantom{|}}=\overline{\check\varphi}\), it follows that
involutions are respected.

Now, let 
\(\pi\) be an irreducible, non trivial representation of \(\mathscr R\). 
For every \(\varphi\in\mathscr B\), set
\(\tilde\pi(\varphi)=\pi(E_-\varphi)\oplus\pi(E_+\varphi)\). From the preceding
remark it follows that \(\tilde\pi\) is a non trivial
representation of \(\mathscr B\). Hence we have \(\pi\circ E_\pm\simeq\pi_+\).
\end{proof}

\begin{proposition}\label{prop:compact}
Let \(\pi_r\) be the representation of the radial algebra \(\mathscr R\)
described in proposition \ref{radial_alg_rep}. Then \(\pi_r(\mathscr R)\)
is norm-dense in the C*-algebra
\(\mathcal K\) of compact operators on a separable, infinite dimensional 
Hilbert space.

In particular, for every 
\(f\in L^1(\mathbb R^2)\cap\widehat{L^1(\mathbb R^2)}\),
\(f(T,R)\) is a compact operator.
\end{proposition}
\begin{proof}
By \cite[VI.12(a)]{reed-simon_1}
it is sufficient to show that
\(\pi_r(\varphi)\) is compact for every \(\varphi\) in some total
subset of \(\mathscr R\); for example if 
\(\varphi=\varphi_1\otimes\varphi_2\) with \(\varphi_j\in 
L^1(\mathbb R)\cap L^2(\mathbb R)\) and \(\varphi_2\) even.
Since in general \(\|\pi_r(\varphi)\|\leqslant\|\varphi\|_{L^1}\), by 
\cite[VI.22(e)]{reed-simon_1} it is sufficient to show 
that there exists a sequence \((\varphi_n)\) in \(L^1\) such
that \(\lim_n\|\varphi_n-\varphi\|_{L^1}=0\) and \(\pi_r(\varphi_n)\)
is Hilbert-Schmidt. Let \(j\in\mathcal C_b(\mathbb R)\) be a bounded, positive
continuous function with \(\lim_{r\rightarrow 0}j(|r|)/\sqrt{|r|}=1\), 
\(0\leqslant j(|r|)\leqslant 1\), 
and \(j(r)=1,r>1\). We define 
\(f_n(t,r)=j(n|r|)\check\varphi(t,r)\).
By construction \(\lim_n\|f_n-\check\varphi\|_{L^\infty}=0\). The functions
\(f_n,\check\varphi\) are continuous and \(L^1\), hence the preceding remark 
implies \(\lim_n\|f_n-\check\varphi\|_{L^1}=0\).  Since the Fourier
transform is bounded as a linear map from \(L^1\) to \(L^\infty\)
(see e.g. \cite[Theorem IX.8]{reed-simon_1}), we also have
\(\lim_n\|\hat f_n-\varphi\|_{L^\infty}=0\), which again implies 
\(\lim_n\|\hat f_n-\varphi\|_{L^1}=0\). 
Hence we have our candidate sequence
\(\varphi_n=\hat f_n\). Finally,
\[
\int_{r>0} \frac{dt\;dr}{r}|f_n(t,r)|^2\leqslant\|\varphi_1\|_{L^2}
\left(\int_{0<r<\frac1n}\frac{dr}{r}|j(nr)\check\varphi_2(r)|^2+
n\|\varphi_2\|^2_{L^2}\right)<\infty.
\]
Hence each \(\pi_r(\varphi_n)=f_n(T,R)\) is Hilbert-Schmidt by lemma 
\ref{lm:HS}.
\end{proof}

\section{Cartesian vs Radial Weyl Quantisation}
\label{sec:cart_ves_radial}
In this section \(T,\bm X,\bm C,R\) are the operators of the universal
representation described in proposition \ref{def:univ_rep}, 
(we drop the apex \(u\), for simplicity). An easy generalisation 
of the argument of proposition \ref{prop:eureka} gives 
the Weyl operators 
\begin{equation}
e^{i(\alpha T+\bm{\beta X})}=
e^{i\alpha T}e^{i\frac{e^{\alpha}-1}{\alpha}\bm{\beta X}},
\quad(\alpha,\bm\beta)\in\mathbb R\times\mathbb R^d
\end{equation}
in \(d+1\) dimensions.

The Weyl operators, by definition, are the quantum replacement of Fourier
characters. Hence we mimic Weyl proposal \cite{weyl:1928}:
\begin{definition}
The Cartesian quantisation of the function 
\(f\in L^1(\mathbb R^{d+1})\cap \widehat{L^1(\mathbb R^{d+1})}\) 
is the operator 
\begin{equation}
f(T,\bm X)=\frac1{\sqrt{(2\pi)^{d+1}}}\int d\alpha\; d\bm\beta\;
\hat f(\alpha,\bm\beta)e^{i(\alpha T+\bm{\beta X})}.
\end{equation}
\end{definition}

We may define the star product by setting
\begin{equation}
(f\star g)(T,\bm X)=f(T,\bm X)g(T,\bm X),
\end{equation}
where the operator product is taken on the right hand side, and 
injectivity of the Weyl quantisation (with universal Weyl operators) is
used.

The above definition does not emphasise the radial nature of 
the quantisation. We give an alternative definition of quantisation
for an alternative class of symbols, called radial symbols, and we will compare the two descriptions.
\begin{definition}
Let \({\mathcal F}\in\mathcal C(S^{d-1}\sqcup\{\bm 0\},L^1(\mathbb R^2))\) a continuous
function of \(S^{d-1}\sqcup\{\bm 0\}\) with values in \(L^1(\mathbb R^2)\).
If, for any \(\bm c\in S^{d-1}\sqcup\{\bm 0\}\) fixed, we also have
\(\widehat{\mathcal F(\bm c)}\in L^1(\mathbb R^2)\), then
we call \({\mathcal F}\) a radial symbol. 

A radial symbol \(\mathcal F\) is said continuous at zero if\/ 
\(\lim_{r\rightarrow 0}\mathcal F(\bm c(r))(t,r)\)
exists and does not depend on the particular choice of the continuous 
\(S^{d-1}\)-valued function
\(\bm c(r)\) of \((0,\infty)\).

For any radial symbol \({\mathcal F}\), we define
\begin{subequations}
\begin{equation}
{\mathcal F}(\bm c)(T,R)=\frac1{2\pi} \int \widehat{\mathcal F}(\bm c)(\alpha,\beta)e^{i(\alpha T+\beta R)}.
\end{equation}

Finally we define the (universal) 
radial quantisation \({\mathcal F}(\bm C)(T,R)\) 
of \({\mathcal F}\) by means of
the obvious generalisation of the continuous functional calculus: with 
\(P(d\bm c)\) the joint spectral measure of the pairwise commuting operators 
\((C_1,\dotsc,C_n)\),
\begin{equation}
{\mathcal F}(\bm C)(T,R)=\int {\mathcal F}(\bm c)(T,R)P(d\bm c).
\end{equation}
\end{subequations}
\end{definition}
Note that the definition is well posed, since each \(C_j\) strongly commutes
with both \(T\) and \(R\). In particular joint spectral projections for 
\(\bm C\) commute with \({\mathcal F}(\bm c)(T,R)\).

The following remarks should not come as a surprise.
\begin{lemma}\label{prop:rad_cart_corresp}
For every \(f\in L^1(\mathbb R^{d+1})\cap\widehat{L^1(\mathbb R^{d+1})}\),
\[
\mathcal F^{f}(\bm c)(t,r)=f(t,r\bm c)
\] 
is a radial symbol continuous at zero, such that 
\begin{equation}\label{eq:ciprovo}
\mathcal F^f(\bm C)(T,R):=f(T,\bm X).
\end{equation}
\end{lemma}
Note that the above establishes a bijective correspondence between
cartesian symbols and radial symbols continuous at zero; in particular
to every  \(\mathcal F\) there corresponds \(f(t,\bm x)=
\mathcal F( |\bm x|^{-1}\bm x)(t,|\bm x|)\). 
\begin{proof}
By reduction theory, it is sufficient to prove
\eqref{eq:ciprovo}
in the irreducible case, where \(\bm C=\bm cI\) and \(\bm X^{(\bm c)}=\bm cR\).
Let \(S\in O(3)\) be such that \(S\bm e_d=\bm c\).
With the change of integration variables 
\(\bm\beta'=S^t\bm\beta\), we get
\[
f(T^{(\bm c)},\bm X^{(\bm c)})=\frac1{2\pi}\int d\alpha\;d\bm\beta\;
\hat f(\alpha,\bm\beta)e^{-i(\alpha T^{(\bm c)}+(\bm{\beta c})R)}=
\int\mathcal F(\bm c)(T^{(\bm c)},R^{(\bm c)}),
\]
where
\[
\widehat{\mathcal F(\bm c)}(\alpha,\beta_d')=\int 
d\beta_1'\dotsm d\beta_{d-1}'
\hat f(\alpha, S\bm\beta').
\]
Standard computations then yield \(\mathcal F=\mathcal F^f\).
\end{proof}

Of course, there is a radial star product defined by
\begin{equation}
{\mathcal F_1}(\bm C)(T,R){\mathcal F_2}(\bm C)(T,R)=
(\mathcal F_1\medstar\mathcal F_2)(\bm C)(T,R)
\end{equation}
which is intertwined with the cartesian star product by the above
correspondence:
\begin{equation}
\mathcal F^f\medstar\mathcal F^g=\mathcal F^{f\star g}.
\end{equation}

Although \(f(T,\bm X)=\mathcal F^f(\bm C)(T,R)\) 
is not compact any more because
of the amplification with infinite multiplicity, we still may define an 
unbounded trace on the algebra of radial symbols by
\begin{equation}
\tau_r(\mathcal F)=\int_{S^{d-1}}d\bm c\;\int\limits_{r>0} dt\;dr\frac1r
\mathcal F(\bm c)(t,r)=\int_{S^{d-1}}d\bm c\;\Tr(\mathcal F(\bm c)(T,R),
\end{equation}
which is finite on the symbols \(\mathcal F\) with \(\mathcal F/r\) is
summable.

It follows from the definition and proposition \ref{prop:irred_trace},
\begin{equation}
\tau_r(\mathcal F_1\medstar\mathcal F_2)=
\tau_r(\mathcal F_2\medstar\mathcal F_1),\quad
\tau_r(\overline{\mathcal F}\medstar\mathcal F)\geqslant 0.
\end{equation}

The above trace functional can be written also in terms of cartesian symbols:
\[
\tau_c(f)=\frac{2\pi^{d/2}}{\Gamma(d/2)}
\int dt\;d\bm x|\bm x|^{-d}f(t,\bm x),
\]
where \(\Gamma(z)=\int_0^\infty t^{z-1}e^{-t}dt\). 

If we now stick to the covariant representation described in proposition
\ref{prop:cov_coord}, we find
\[
U(g)f(T,\bm X)U(g)^{-1}=\rho_g(f)(T,\bm X),\quad g\in G_d,
\]
where, for \(g=(A,a,\lambda)\),
\begin{equation}\label{eq:action}
\rho_g(f)(t,\bm x)=f(t-a,\lambda^{-1}A^{-1}\bm x).
\end{equation}
The corresponding action on radial symbols
can be easily computed using lemma \ref{prop:rad_cart_corresp}. 

\section{The C*-Algebra and the Large Scale Limit}
\label{sec:C_star}
The most natural and canonical definition of C*-algebra of the 
\(\kappa\)-Minkowski spacetime is the smallest C*-algebra containing
all universal quantisations of cartesian symbols. 
\begin{definition}
Let \((T,\bm X)\) be the universal representation on \(\mathfrak H\).
The C*-algebra of the \(\kappa\)-Minkowski 
spacetime in \(d+1\) dimensions is defined as
\begin{equation}
\mathcal A_d=\{f(T,\bm X):f\in L^1(\mathbb R^{d+1})\cap 
\widehat{L^1(\mathbb R^{d+1})}\}
\overline{\phantom{k}}^{\|\cdot\|_{B(\mathfrak H)}}
\end{equation}
\end{definition}

From the discussion of the preceding section, it is clear that
\(\mathcal A_d\) can be obtained equivalently 
from the universal quantisations of
the radial symbols continuous at zero.

The next lemma will be crucial, 
in that it will both allow for the explicit
characterisation of \(\mathcal A_d\), and to dismiss the condition of 
continuity at zero of radial symbols. The two facts are related; the latter
will also be responsible of the exotic topology of the large scale limit.

\begin{lemma}
Let \(\mathcal C(S^{d-1},L^1(\mathscr R))\) be the Banach space
of continuous functions \(\Phi:S^{d-1}\mapsto \mathscr R \) 
with values in the radial algebra
\(\mathscr R\), and with norm
\[
\|\Phi\|=\sup_{\bm c\in S^{d-1}}\|\Phi(\bm c)\|_{\mathscr R}
\]
We denote by \(\hat{\mathcal F}\) the fibrewise Fourier transform on radial 
symbols, namely
\[
\hat{\mathcal F}(\bm c)(\alpha,\beta)=
\widehat{\mathcal F(\bm c)}(\alpha,\beta)=\frac1{2\pi}\int dt dr
\mathcal F(\bm c)(t,r)e^{-i(\alpha t +\beta r)}.
\]
Then the set of all fibrewise Fourier transforms of radial symbols continuous
at zero is dense in \(\mathcal C(S^{d-1},L^1(\mathscr R))\).
\end{lemma}
\begin{proof}
Since removing the request of continuity at zero the set 
\(\{\hat{\mathcal F}\}\)
is dense, it is sufficient to show that, for any radial symbol
\(\mathcal F\) (possibly not continuous at zero), there exists a sequence
\((\mathcal F_n)\) of radial symbols continuous at zero such that
\(\lim_n\|\hat{\mathcal F}-\hat{\mathcal F}_n\|=0\). It is sufficient to 
further restrict ourselves to the case where \(\mathcal F\) is of the form
\(\mathcal F(\bm c)(t,r)=f(\bm c)g(t,r)\), since they provide, through
Fourier transformation, a total set in \(\mathcal C(S^{d-1},L^1(\mathscr R))\).
Let \(j\) be defined as in the proof of proposition \ref{prop:compact},
and \(\mathcal F_n(\bm c)(t,r)=f(\bm c)j(n|r|)(g(t,r)\). Then the proof
follows by the same arguments used to prove proposition \ref{prop:compact}.
\end{proof}

Combining the preceding lemma with propositions \ref{radial_alg_rep},
\ref{prop:compact}, we see that the Banach *-algebra 
\(\mathcal C(S^{d-1},\mathscr R)\) (equipped with pointwise product
in \(\mathscr R\)) 
has a unique C*-completion \(\mathcal C(S^{d-1},\mathcal K)\). As a consequence,
\begin{proposition}
The C*-algebra \(\mathcal A_d\) is isomorphic with 
\(\mathcal C(S^{d-1},\mathcal K)\oplus \mathcal C_\infty(\mathbb R)\).
The action \eqref{eq:action} of \(G_d\) 
extends to an action by automorphism of \(\mathcal A_d\),
still denoted by \(\rho\).
\end{proposition}
Note that the commutative 
component \(\mathcal C_\infty(\mathbb R)\) shows up in the universal 
C*-completion because of the presence of continuously many
equivalence classes of irreducible trivial representations.

We may then extend the quantisation of a function 
\(f\in\mathcal C(\overset{\circ}{M}{}^{(1,d)}\sqcup \mathbb R)\simeq 
\mathcal C(\overset{\circ}{M}{}^{(1,d)})\oplus\mathcal C_\infty(\mathbb R)\)
according to the above proposition. 
It is then clear that the large \(\kappa\) limit (which is the same
as the large scale limit), can be performed separately in the trivial
and non trivial component and gives 
\(\mathcal C(\overset{\circ}{M}{}^{(1,d)}\sqcup \mathbb R))\).

\section{Remarks on the Uncertainty Relations}
\label{sec:uncertainty}

We now consider the Heisenberg uncertainty relations arising from the 
commutation relations between time \(T\) and radial distance \(R>0\),
which read
\begin{equation}
c\Delta_\omega(T)\Delta_\omega(R)\geqslant\frac{c}{2}\omega(|[T,R]|)=
\frac{1}{2\kappa}\omega(R).
\end{equation}
for any state \(\omega\)in the domain of \(T,R\), where \(c\) is light speed. 
In this section, we restore the explicit dependence on \(c,\kappa\).

We may observe that, for a state whose expectation on \(X\) is 
very close to the origin (namely \(\omega(R)\) very ``small''), 
the above uncertainty 
relations tend to have little content.  Indeed, 
\begin{lemma}\label{lemma:uncertanties}
For every \(\varepsilon,\eta>0\) there exists a non trivial 
pure vector state in the domain of \(T,R\), such that
\[
\Delta_\omega(T)<\varepsilon,\quad \Delta_\omega(X)<\eta.
\] 
The state can be chosen to belong to the non trivial component.
\end{lemma}

This means that there is no limit on the precision with which
we may simultaneously localise all the spacetime coordinates, 
at least in the region close to the space origin. 
This is in plain contrast with the standard 
motivations for spacetime quantisation, namely to prevent the formation
of closed horizons as an effect of localisation {\itshape alone} (see 
\cite{Doplicher:1994tu}).The proof will be found at the end of this section.

The above is yet another way of saying that the model is approximately 
commutative close to the origin, while noncommutative effects grow 
the more 
important, the more the states we consider are 
localised far away from the origin. This suggests that, on the other side,
it could be of some interest to have estimates about 
how fast noncommutativity grows with distance. 

Indeed, for a state localised at distance  \(\omega(R)=L\) from the origin, 
we may rewrite the uncertainty relations as
\begin{equation}
L\leqslant 2\kappa c\Delta T\Delta R
\end{equation}
(note that \(\kappa>0\)).

We may ask ourselves at which distance from the centre
of \(\kappa\)-Minkowski is is meaningful to speak of strong interactions. 
Asking for
\(\Delta R\sim c\Delta T\ll 10^{-19}m\) we find (if \(\kappa\sim 10^{35}m\)) 
\begin{equation}
L\ll 10^{-3}m.
\end{equation}
The peak nominal beam size at LHC is \(1.2\cdot 10^{-3}\;\)m \cite[2.2.2]{Evans:2008zzb}.

The diametre of the classical orbit of the electron in the fundamental
state of the Hydrogen atom is around \(\ell_0=.5\cdot 10^{-10}m\); 
the classical period of the orbit is 
\(\tau_0\sim 1.5\cdot 10^{-16}m\). 
We take \(\tau_0,\ell_0\) as characteristic scales of
atomic physics, and we ask for the avalability of 
localisation states which
are not so undetermined to destroy the meaning of atomic physics, namely 
\(\Delta T\ll \tau_0,\Delta R\ll\ell_0\);
we find
\begin{equation}
L\ll \kappa c\tau_0\ell_0\sim \kappa \cdot (10^{-18}m). 
\end{equation}
With \(\kappa\sim 10^{35}m\), namely of order of the inverse 
Planck length, we find
\begin{equation}
L\ll 10^{17}m\sim 10\; \text{light--years}.
\end{equation}
The Galaxie is \(10^3\) light--years thick; \(\alpha\)-Centauri is
five light--years from Earth.

We could also turn things the other way round. We may ask LHC physics
to exist and be the same no matter 
where the LHC is built on Earth. Since the diametre of the Earth is 
\(L\sim 10^7m\), the condition \(c\Delta T\Delta R\ll 10^{-38}m^2\) gives
\begin{equation}
\frac 1\kappa\ll 10^{-45}m
\end{equation}
which is a less than a billionth of the Planck length. 

Of course, it could be objected to such estimates that the model only
should be intended ``locally'', and that large distance effects could be 
``cut off''. It is however not clear how to do this. Commutation relations
are global, and the unavailability of a suitable generalisation of the 
concept of locality is precisely the obstruction preventing us from
going beyond semiclassical models of flat spacetime, which are globally 
defined. For the model as it stands, the above estimates are meaningful.

We now come to the discussion which is summarised by lemma
\ref{lemma:uncertanties} above.

For any state which is pure and belongs to 
the trivial component (via GNS), both \(T,R\) have null uncertainty,
since the trivial component is commutative. Hence the lemma is trivially true
in this case.

It may seem however that the above is a pathology due to the special status
of the origin. What about states localised 
``very close'' to the 
origin, but not precisely there?   
Let \(\varepsilon,\eta>0\) be any arbitrary choice of (``small'')
positive numbers. There always is
a choice of \(\xi^{\varepsilon}\in\mathscr D(P)\) derivable and 
with compact support such that \(\Delta_{\xi^\varepsilon}(P)<\varepsilon\). 
For such a choice, let 
\(\xi^\varepsilon_\lambda(s)=(e^{-i\lambda P}\xi^\varepsilon)(s)=\xi^\varepsilon(s-\lambda)\); since unitary transformations preserve uncertainties,
\(\Delta_{\xi^\varepsilon_\lambda}(P)<\varepsilon\) for every 
\(\lambda\). But we also have \(\xi_\lambda^\varepsilon\in\mathscr D(e^{-Q})\);
moreover, as a consequence of the compactness of the support,
\[
\lim_{\lambda\rightarrow\infty} \Delta_{\xi^\varepsilon_\lambda}(e^{-Q})=0.
\]
In particular, there is a \(\lambda_\eta\) such that
\(\Delta_{\xi^\varepsilon_{\lambda_\eta}}(e^{-Q})<\eta\).
We found a state not beloging to the trivial component, and such that
\[
\Delta(T)<\varepsilon,\quad \Delta(R)<\eta.
\]

\section{Conclusions and Outlook}

On the mathematical side, we found that the C*-algebra of the
\(\kappa\)-Minkowski model and its representations can
be discussed thoroughly, leading to a sound quantisation prescription,
which is canonically associated with  the abstract Lie 
algebra underlying the relations. 

On the side of interpretation, on the contrary, we observed some
features which (together with other evident remarks which we also collect here for completeness) are not fully satisfactory from the point of view of spacetime
quantisation. 
\begin{itemize}
\item The main motivation for 
spacetime quantisation, namely to prevent
arbitrarily precise localisation (which could lead to horizon formation)
is lost for this model. 
\item Covariance under Lorentz boosts is severely broken.
\item Moreover, translation covariance is so severely broken that the origin
of space, which already got a special status from the relations, remains 
classical at \(\kappa\neq 0\), and remains topologically disjoint from the 
rest in the large scale limit.
\item
The model is classical in the origin and grows noncommutative very fastly
as the distance from the origin increases. 
\end{itemize}

We also have seen that the C*-algebra and the quantisation prescription can be
derived from the canonical CCR quantisation. In particular, this means that
deformed Lorentz covariance could be easily established by exploiting the 
twisted covariance of \cite{Chaichian:2004za,Wess:2003da}. We discuss this feature
in \cite{Dabrowski:2009mw}, where we also show that the 
deformed--covariant \(\kappa\)-Minkowski spacetime 
can be obtained as a non invariant
restriction of a fully covariant model (thus reproducing the situation 
described in \cite{Piacitelli:2009tb,Piacitelli:2009fa}). 
The fully covariant model will be 
obtained as a minimal central covariantisation of the usual
\(\kappa\)-Minkowski model. This will hopefully 
shed some light on noncommutative
covariance, but unfortunately will not cure the lack of stability of spacetime
under localisation alone, which will survive covariantisation. By the same 
techniques we may obtain a fully Poincar\'e covariant model. Once again,
the initial \(\kappa\)-Minkowski is contained in the covariantised
model as a subrepresentation; hence the states with sharp localisation 
still will be available, and will be localisable everywhere.

\footnotesize
\bibliographystyle{utphys}
\bibliography{mybibtex}

\end{document}